\documentclass[letterpaper, 10 pt, conference]{ieeeconf}  

\IEEEoverridecommandlockouts                              
\usepackage{amsmath} 
\usepackage{amssymb}  
\usepackage[all]{xy}
\usepackage{algorithm}
\usepackage{algpseudocode}
\usepackage{cite}
\usepackage{flushend}
\usepackage{alltt}
\usepackage{todonotes}
\usepackage{graphicx}
\usepackage{caption}
\usepackage{subcaption}
\usepackage{soul,xcolor}
\setstcolor{red}

\usepackage{verbatim}

\newtheorem{defn}{Definition}
\newtheorem{thm}{Theorem}
\newtheorem{prop}{Proposition}
\newtheorem{cor}{Corollary}
\newtheorem{lemma}{Lemma}
\newtheorem{exm}{Example}
\newtheorem{remark}{Remark}
\newtheorem{experiment}{Experiment}

\def\Re{\mathbb{R}}
\DeclareMathOperator{\rank}{rank}
\DeclareMathOperator{\spark}{spark}
\DeclareMathOperator{\nullsp}{Null}
\DeclareMathOperator{\nullity}{nullity}
\DeclareMathOperator{\supprt}{supp}

\DeclareMathOperator*{\argmin}{arg\,min}
\DeclareMathOperator*{\Ext}{Ext}

\title{\LARGE \bf Sparse Recovery over Graph Incidence Matrices}

\author{ 
Mengnan Zhao, M. Devrim Kaba, Ren\'e Vidal, Daniel P. Robinson, and Enrique Mallada
\thanks{M. Zhao and E. Mallada are with the Department of Electrical and Computer Engineering, M.D. Kaba and D.P. Robinson are with the Department of Applied Mathematics and Statistics, and R. Vidal is with the Department of Biomedical Engineering, at The Johns Hopkins University, Baltimore, MD 21218, USA. Emails: {\tt \small \{mzhao21, mkaba1, rvidal, daniel.p.robinson, mallada\}@jhu.edu}}%
\thanks{This work was supported by NSF grants 1618637 and AMPS 1736448.}
}

\usepackage{ifthen} 
\newboolean{showcomments}
\setboolean{showcomments}{true}   
\usepackage{todonotes}
\usepackage{marginnote}

\definecolor{bleudefrance}{rgb}{0.19, 0.55, 0.91}
\definecolor{ao(english)}{rgb}{0.0, 0.5, 0.0}

\newcommand{\enrique}[1]{  \ifthenelse{\boolean{showcomments}}
{\todo[inline,color=bleudefrance]{Enrique says: #1}}{}}

\newcommand{\emmargin}[1]{\ifthenelse{\boolean{showcomments}}{\marginpar{\color{bleudefrance}\tiny EM: #1}}{}}

\newcommand{\addcite}[0]{\ifthenelse{\boolean{showcomments}}{\textcolor{purple}{(add cite(s)) }}{}}%

\begin{document}

\maketitle

\begin{abstract}
Classical results in sparse recovery guarantee the exact reconstruction of $s$-sparse signals under assumptions on the dictionary that are either too strong or NP-hard to check. Moreover, such results may be pessimistic in practice since they are based on a worst-case analysis. In this paper, we consider the sparse recovery of signals defined over a graph, for which the dictionary takes the form of an incidence matrix. We derive necessary and sufficient conditions for sparse recovery, which depend on properties of the cycles of the graph that can be checked in polynomial time. We also derive support-dependent conditions for sparse recovery that depend only on the intersection of the cycles of the graph with the support of the signal. Finally, we exploit sparsity properties on the measurements and the structure of incidence matrices to propose a specialized sub-graph-based recovery algorithm that outperforms the standard $\ell_1$-minimization approach.
\end{abstract}

\section{INTRODUCTION}
Recently, sparse recovery methods (e.g. see \cite{donoho2006compressed,mallat2008wavelet}) have become very popular for compressing and processing high-dimensional data. In particular, they have found widespread applications in data acquisition \cite{candes2008introduction}, machine learning \cite{elhamifar2009sparse,wright2009robust}, medical imaging \cite{lustig2006rapid, lustig2007sparse,candes2006robust,LiMiZoSe:15}, and networking \cite{coates2007compressed, haupt2008compressed, xu2011compressive}. 

The goal of sparse recovery is to reconstruct a signal $\bar{x}\in \mathbb{R}^n$ from $m<n$ linear measurements $y=\Phi \bar{x} \in \mathbb{R}^m$, where $\Phi \in \mathbb{R}^{m\times n}$ is the measurement matrix (also known as dictionary or sparsifying basis). In general, the recovery problem is ill-posed unless $\bar{x}$ is assumed to satisfy additional assumptions. For example, if we assume that $\bar{x}$ is $s$-sparse\footnote{The vector $\bar{x}$ is $s$-sparse if and only if at most $s$ of its entries are nonzero.}, where $s \ll n$, then mild conditions on the measurement matrix $\Phi$ (see Lemma \ref{lem:sparkCondition}) allow for the recovery $\bar{x}$ from $y$ using the $\ell_0$-minimization problem
\begin{equation} \label{eq:l0} 
\min_{x} \  \|x\|_0  \text{ s.t. }  y=\Phi x,
\end{equation}
where $\|x\|_0$ denotes the number of nonzero entries in $x$. However, problem~\eqref{eq:l0} is known to be NP-hard \cite{michael1979computers}. To address this challenge, a common strategy is to solve a convex relaxation of \eqref{eq:l0} based on $\ell_1$-minimization given by
\begin{equation} \label{eq:l1}
\min_x \|x\|_1 \ \text{ s.t. } \ y=\Phi x,
\end{equation}
which can be written equivalently as a linear program.  It is known that the sparse signal can be recovered from \eqref{eq:l1} if the measurement matrix $\Phi$ satisfies certain conditions. In general, these conditions are either computationally hard to verify, or too conservative so that false negative certificates are often encountered. For example, the \emph{Nullspace Property} (NUP)\cite{cohen2009compressed}, which provides necessary and sufficient conditions for sparse recovery, and the \emph{Restricted Isometry Property} (RIP) \cite{candes2005decoding,baraniuk2008simple}, which is only a  sufficient condition for sparse recovery, are both NP-hard~\cite{TiPf:14} to check. On the other hand, the \emph{Mutual Coherence}~\cite{donoho2001uncertainty} property is a sufficient condition that can be efficiently verified\cite{TiPf:14}, but is conservative in that it fails to certify that sparse recovery is possible for many measurement matrices seen in practice. 

Another major limitation of these recovery guarantees and their associated computational complexity arises because they must allow for the worst case problem in their derivation. For example, the fact that the NUP and RIP are NP-hard to check does not prohibit efficient verification for particular subclasses of matrices. Similarly, even when the NUP is not satisfied for a sparsity level $s$, it may still be possible to recover certain subsets of $s$-sparse signals by exploiting additional knowledge about their support. These observations suggest studying the sparse recovery problem for subclasses of matrices and signals to obtain conditions that are easier to verify as well as specialized algorithms with improved recovery performance.
The goal of this paper is to study sparse recovery of signals that are defined over graphs when the measurement matrix is the graph's incidence matrix. 
Our interest in incidence matrices stems from the fact that they are a fundamental representation of graphs, and thus a natural choice for the dictionary when analyzing network flows. In various application areas like communication networks, social networks, and transportation networks, the incidence matrices naturally appear when modeling the flow of information, disease, and goods (e.g., the detection of sparse structural network changes via observations at the nodes can be modeled as (1), where the incidence matrix serves as a measurement matrix).

The main contributions of this paper are as follows:
\begin{enumerate}
\item We derive a topological characterization of the NUP for incidence matrices. Specifically, we show that the NUP for these matrices is equivalent to a condition on simple cycles of the graph, which is a finite subset of the nullspace of the incidence matrix. As a consequence, we show that for incidence matrices the sparse recovery guarantee depends only on the girth\footnote{Girth: Size of the smallest cycle.} of the underlying graph. This overcomes NP-hardness, as the girth of a graph can be calculated in polynomial time.

\item Using the above topological characterization, we further derive necessary and sufficient conditions on the support of the sparse signal that enable its recovery. Specifically, for incidence matrices we show that all signals with a given support can be recovered via \eqref{eq:l1} if and only if the support consists of strictly less than half of the edges in every simple cycle of the graph. Since our conditions on $\bar{x}$ depend on its support, we will refer to them as \emph{support-dependent} conditions.

\item We propose a specialized algorithm that utilizes the knowledge of the support of the measurements and the structure of incidence matrices to constrain the support of the signal $\bar{x}$, and consequently can guarantee sparse recovery of $\bar{x}$ under even weaker conditions.
\end{enumerate}

The remainder of this paper is organized as follows: In Section II we review basic concepts from convex analysis, graph theory, and sparse recovery. In Section III we formulate sparse recovery problem for incidence matrices, derive conditions for sparse and support-dependent recovery, and propose an efficient algorithm to solve the problem. In Section IV we present numerical experiments that illustrate our theoretical results, and in Section V we provide some conclusions.

\section{PRELIMINARIES}

\subsection{Notation}
Given $x\in\Re^n$, we let $\|x\|_p$ for $p\geq 1$ denote the $\ell_p$-norm, and denote the function that counts the number of nonzero entries in $x$ by $\|x\|_0$. Although, the latter function is not a norm, we follow the common practice of calling it the $\ell_0$-norm. 
For $p\geq 1$, we define the unit $\ell_p$-sphere in $\mathbb{R}^n$ as $\mathbb{S}_p^{n-1} := \{x\in \mathbb{R}^n : \|x\|_p=1\}$.
Similarly, the unit $\ell_p$-ball in $\Re^n$ is defined as 
$\mathbb{B}_p^{n} := \{x\in \mathbb{R}^n : \|x\|_p\leq 1\}.$

The \emph{nullspace} of a matrix $\Phi$ will be denoted by $\nullsp(\Phi)$. 
As usual, $|S|$ denotes the cardinality of a set $S$. For a vector $x\in \mathbb{R}^n$ and a nonempty index set $S\subseteq \{1,\dots, n\}$, we denote the subvector of $x$ that corresponds to $S$ by $x_S$, and the associated mapping by $\Gamma_S : \mathbb{R}^n \to \mathbb{R}^{|S|}$, i.e., 
$x_S:= \Gamma_S(x)$. The complement of $S$ in $\{1,\dots, n\}$ is denoted by $S^c$.

\subsection{Convex analysis}
Critical to our analysis will be the notion of extreme points of a convex set and of quasi-convex functions.
\begin{defn}
An \emph{extreme point} of a convex set $C\subset \mathbb{R}^n$ is a point in $C$ that cannot be written as a convex combination of two different points from $C$. The set of all extreme points of a convex set $C \subset \mathbb{R}^n$ is denoted by $\Ext (C)$.  
\end{defn}

\begin{defn}
Let $C\subseteq \mathbb{R}^n$ be a convex set. A function $f:C\to \mathbb{R}$ is called \emph{quasi-convex} if and only if for all $x,y \in C$ and $\lambda\in [0,1]$ it holds that
\begin{equation}
f(\lambda x + (1-\lambda)y) \leq \max\{f(x), f(y)\}.
\end{equation}
\end{defn}

It is easy to show that a function $f:C\to \mathbb{R}$ is quasi-convex if and only if every sublevel set 
$S_\alpha:=\{x\in C \,| \, f(x)\leq \alpha\}$
is a convex set. Every convex function is quasi-convex. In particular, $\ell_p$-norm functions for $p\geq 1$ are quasi-convex. The next  result on quasi-convex functions~\cite{FlBa:15} is included for completeness.

\begin{prop}\label{prop:qconvmax}
Let $C\subset \mathbb{R}^n$ be a compact convex set and $f:C\to \mathbb{R}$ be a continuous quasi-convex function. Then $f$ attains its maximum value at an extreme point of $C$. 
\end{prop}

\subsection{Graph theory}
A directed graph with vertex set $V=\{v_1, \dots, v_m\}$ and edge set $E=\{e_1,\dots, e_n\}\subseteq V\times V$ will be denoted by $\mathcal{G}(V,E)$.
When the edge and vertex sets are irrelevant to the discussion, we will drop them from the notation and denote the graph by $\mathcal{G}$. Sometimes the vertex set and the edge set of $\mathcal{G}$ will be denoted by $V(\mathcal{G})$ and $E(\mathcal{G})$ respectively. A graph $\mathcal{G}$ is called \emph{simple} if there is at most one edge connecting any two vertices and no edge starts and ends at the same vertex, i.e. $\mathcal{G}$ has no self loops. Henceforth, $\mathcal{G}$ will always denote a simple directed graph with a finite number of edges and vertices. 
Although we focus on directed graphs, our analysis only requires the undirected variants of the definitions for paths, cycles, and connectivity.

We say that two vertices $\{a,b\}\subseteq V$ are \emph{adjacent} if either $(a,b)\in E$ or $(b,a)\in E$. A sequence $(u_1,\dots, u_r)$ of distinct vertices of a graph $\mathcal{G}$ such that $u_i$ and $u_{i+1}$ are adjacent for all $i\in\{1,\dots,r-1\}$ is called a \emph{path}. A \emph{connected component} $ \hat{\mathcal{G}}$ of $\mathcal{G}$ is a subgraph of $\mathcal{G}$ in which any two distinct vertices are connected to each other by a path and no edge exists between $V(\hat{\mathcal{G}})$ and $V(\mathcal{G})\setminus V(\hat{\mathcal{G}})$. We will say that the graph $\mathcal{G}$ is \emph{connected} (for directed graphs this is often referred to as weakly connected) if $\mathcal{G}$ has a unique connected component.
A sequence $(u_1,\dots, u_r, u_1)$ of adjacent vertices of a graph $\mathcal{G}$ is called a \emph{(simple) cycle} if $r\geq 3$ and $u_i\neq u_j$ whenever $i\neq j$; the  length of such a cycle is $r$. 
The length of the shortest simple cycle of a graph $\mathcal{G}$ is called the \emph{girth} of $\mathcal{G}$. For an acyclic graph (i.e., a graph with no cycles), the girth is defined to be $+\infty$. Since acyclic graphs are not interesting for our purposes, we will assume that the girth is finite. That is, the graph has at least one simple cycle.

Associated with a directed graph $\mathcal{G} = \mathcal{G}(V,E)$, we can define the incidence matrix $A=A(\mathcal{G}) \in \mathbb{R}^{m \times n}$ as
\begin{equation}
a_{ij}=
\begin{cases}
-1, &\text{ if } v_i \text{ is the initial vertex of edge } e_j, \\
\phantom{-}1, &\text{ if } v_i \text{ is the terminal vertex of edge } e_j,\\
\phantom{-}0, &\text{ otherwise.}
\end{cases}
\end{equation}
For a nonempty index set $S\subseteq \{1, \dots, n\}$, the subgraph of $\mathcal{G}$ consisting of edges $\{e_j\mid j\in S\}$ is denoted by $\mathcal{G}_S$. The incidence matrix of $\mathcal{G}_S$ is denoted by $A_S$.
Let $\mathcal{C} = (u_1,\dots, u_r, u_{r+1}=u_1)$ be a simple cycle of a simple directed graph $\mathcal{G}$. Then, $\mathcal{C}$ can be associated with a vector $w(\mathcal{C})\in \mathbb{R}^n$, where each coordinate $w_j$ of $w(\mathcal{C})$ is defined as
\begin{equation}\label{def.w}
w_j=
\begin{cases}
\phantom{-}1, &\text{ if } e_j = (u_i,u_{i+1})\text{ for some }i\in \{1,\dots, r\}\\
-1, &\text{ if } e_j= (u_{i+1}, u_{i})\text{ for some }i\in \{1,\dots, r\}\\
\phantom{-}0, &\text{ otherwise.}
\end{cases}
\end{equation}
We now define the \emph{cycle space} of $\mathcal{G}$ as the subspace spanned by $\{w(\mathcal{C}) \mid \mathcal{C}\text{ is a simple cycle of } \mathcal{G}\}$.\footnote{Sometimes this subspace is called the \emph{Flow Space} \cite{GoRo:01}.}
\begin{remark}\label{remark.cs.is.ns}
The cycle space of $\mathcal{G}$ is exactly the nullspace of the incidence matrix $A(\mathcal{G})$ \cite{GoRo:01}.
\end{remark}
\begin{remark} \label{remark:rank}
The dimension of the nullspace of $A(\mathcal{G})$ is $n-m+k$, where $k$ is the number of connected components of $\mathcal{G}$. Hence, the rank of $A(\mathcal{G})$ is $m-k$ \cite{GoRo:01}.
\end{remark}

\begin{figure}[t]
 \[ \xymatrix{
 & & \xy*{8}*\cir<6pt>{}\endxy \ar[dl]_9 \ar[dr]_8 \\
 \xy*{1}*\cir<6pt>{}\endxy \ar[r]^1 \ar[d]^4 & \xy*{2}*\cir<6pt>{}\endxy \ar[d]^2 & & \xy*{7}*\cir<6pt>{}\endxy \\
 \xy*{4}*\cir<6pt>{}\endxy  & \xy*{3}*\cir<6pt>{}\endxy \ar[l]^3  \ar[r]_5 &  \xy*{5}*\cir<6pt>{}\endxy \ar[r]_6   & \xy*{6}*\cir<6pt>{}\endxy \ar[u]_7  &  \xy*{9}*\cir<6pt>{}\endxy \ar[l]_{10}
 }
  \]
\caption{A connected directed simple graph.}\label{Fig1}
\end{figure}
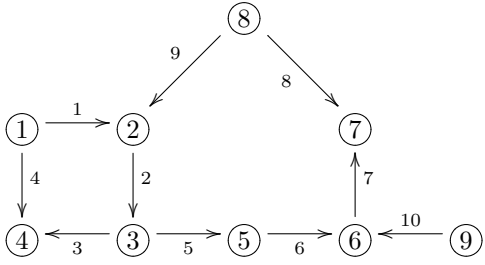


\begin{exm} \label{GraphExample}
The incidence matrix of the graph in Fig. \ref{Fig1} is given by
\begin{equation*}
A=
\begingroup 
\setlength\arraycolsep{3pt}
\begin{bmatrix}
-1 &0 &0 &-1 &0 &0 &0 &0 &0 &0\\
1 &-1 &0 &0 &0 &0 &0 &0 &1 &0\\
0 &1  &-1 &0 &-1 &0 &0 &0 &0 &0\\
0 &0 &1 &1 &0 &0 &0 &0 &0 &0\\
0 &0 &0 &0 &1 &-1 &0 &0 &0 &0 \\
0 &0 &0 &0 &0 &1 &-1 &0 &0 &1\\
0 &0 &0 &0 &0 &0 &1 &1 &0 &0\\
0 &0 &0 &0 &0 &0 &0 &-1 &-1 &0\\
0 &0 &0 &0 &0 &0 &0 &0 &0 &-1
\end{bmatrix},
\endgroup
\end{equation*}
and the vectors from~\eqref{def.w} associated with the simple cycles
\begin{align*}
\mathcal{C}_1 &= (1,2,3,4,1), \\
\mathcal{C}_2 &= (2,3,5,6,7,8,2), \\
\mathcal{C}_3 &= (1,2,8,7,6,5,3,4,1)
\end{align*}
are given by 
\begin{align*}
w_1&=\begin{bmatrix}
1 &1 &1 &-1 &0 &0 &0 &0 &0 &0
\end{bmatrix}^T\\
w_2&=\begin{bmatrix}
0&1 &0 &0 &1 &1 &1 &-1 &1 &0
\end{bmatrix}^T\\
w_3&=\begin{bmatrix}
1 &0 &1 &-1 &-1 &-1 &-1 &1 &-1 &0
\end{bmatrix}^T.
\end{align*}
Note that $w_3=w_1-w_2$, so that $\nullsp(A(\mathcal{G}))$ must be of dimension $2$. The addition/subtraction operations on $w_j$'s correspond to ``addition/subtraction" operations on edges.
\end{exm}

\subsection{Sparse Recovery}


It is known~\cite{FoRa:13} that every vector $\bar{x}\in \mathbb{R}^n$ with $\supprt(\bar{x})\subseteq S$ is the unique solution to \eqref{eq:l1}
if and only if 
\begin{equation}\label{eq:ImpliedIneq}
\|\eta_{S}\|_1 < \|\eta_{S^c}\|_1
\ \text{for all} \ \eta\in \nullsp(\Phi)\setminus \{0\}.
\end{equation}
This leads us to the definition of the \emph{Nullspace Property}.
\begin{defn}[Nullspace Property, NUP] \label{def.nup}
A matrix $\Phi\in \mathbb{R}^{m\times n}$ is said to satisfy the \emph{nullspace property} (NUP) of order $s$, if for any $\eta\in \nullsp(\Phi) \setminus \{0\}$, and any nonempty index set $S\subseteq \{1,\dots,n\}$ with $|S|\leq s$, we have
$$\|\eta_S\|_1 < \|\eta_{S^c}\|_1.$$
\end{defn}

Another needed concept is the \emph{spark} of a matrix.

\begin{defn}[Spark]
The spark of a matrix $\Phi\in \mathbb{R}^{m\times n}$ is the smallest number of linearly dependent columns in $\Phi$. Formally, we have
$\spark(\Phi):=\min_{\eta \neq 0} \|\eta\|_0$ s.t. $\Phi \eta=0$.
\end{defn}

We note that the rank of a matrix may be used to bound its spark.  Specifically, it holds that
\begin{equation}\label{eq:sparkineq}
\spark(\Phi) \leq \rank(\Phi)+1.
\end{equation}
The spark may be used to provide a necessary and sufficient condition for uniqueness of sparse solutions \cite{FoRa:13}.
\begin{lemma}\label{lem:sparkCondition}
For any matrix $\Phi\in \mathbb{R}^{m\times n}$, the $s$-sparse signal $\bar{x}\in \mathbb{R}^n$ is the unique solution to the optimization problem
\begin{equation}
\min_{\Phi \bar{x} = \Phi x} \|x\|_0
\end{equation}
if and only if $\spark(\Phi)>2s$.
\end{lemma}

\section{Main Results}\label{sect:mainresults}
It is not difficult to prove that the spark of an incidence matrix $A$ is equal to the girth of the underlying graph $\mathcal{G}$. By combining this fact with Lemma~\ref{lem:sparkCondition}, we obtain the following.
\begin{prop}\label{prop:l0unique}
Let $A = A(\mathcal{G}) \in \mathbb{R}^{m\times n}$ be the incidence matrix of a simple, connected graph $\mathcal{G}$ with girth $g$. Then, for every $s$-sparse vector $\bar{x}\in \mathbb{R}^n$, $\bar{x}$ is the unique solution~to 
\begin{equation}
\min_{A\bar{x}=Ax} \|x\|_0
\end{equation}
if and only if $s < \frac{g}{2}$.
\end{prop}

Even though Proposition~\ref{prop:l0unique} is useful as a uniqueness result, it does not come in handy when one would like to recover the original signal $\bar{x}$ from the measurements $A\bar{x}$. For this purpose, an $\ell_1$-relaxation of the optimization problem is preferable. Hence, one needs a theorem akin to Proposition~\ref{prop:l0unique} that addresses the solutions of the optimization problem
\begin{equation}
\min_{A\bar{x}=Ax} \|x\|_1.\label{eq:mainl1min}
\end{equation}
This leads us to study the NUP for incidence matrices. Specifically, in this section we answer the following questions about the incidence matrix $A\in \mathbb{R}^{m \times n}$ of a simple connected graph $\mathcal{G}$ with $n$ edges and $m$ vertices, and an $s$-sparse vector $\bar{x}\in \mathbb{R}^n$:
\begin{enumerate}
\item What are necessary and sufficient conditions for $A$ to satisfy the NUP of order $s$? Such conditions would guarantee sparse recovery, i.e., that any $s$-sparse signal $\bar{x}$ can be recovered as the unique solution to \eqref{eq:mainl1min}.
\item If traditional sparse recovery is not possible, can we characterize subclasses of sparse signals that are recoverable via \eqref{eq:mainl1min} in terms of the support of the signal and the topology of the graph $\mathcal{G}(V,E)$?
\item Can we use the support of the measurement $A\bar{x}$ and the structure of A to derive constraints on the support of $\bar{x}$ that allow us to modify~\eqref{eq:mainl1min} and successfully recover the sparse signal $\bar{x}$? 
\end{enumerate}

\subsection{Topological Characterization of the Nullspace Property for the Class of Incidence Matrices}
Before addressing the questions above in detail, we would like to build a simple framework to study them. We will start with a reformulation of the NUP.
\begin{lemma}\label{lem:reformNUP}
A matrix $\Phi\in \mathbb{R}^{m\times n}$ satisfies the NUP of order $s$ if and only if
\begin{equation}\label{eq:AltNUP}
\max_{|S|\leq s} \quad \max_{\eta\in \nullsp(\Phi)\cap \mathbb{B}^{n}_1} \|\eta_S\|_1 < \frac{1}{2}.
\end{equation}
\end{lemma}
\begin{proof}
For $\eta\in \nullsp(\Phi)\setminus\{0\}$, using $\|\eta\|_1=\|\eta_S\|_1+\|\eta_{S^c}\|_1$, we get
$$
\|\eta_S\|_1 < \|\eta_{S^c}\|_1
\text{ if and only if } 
\frac{\|\eta_S\|_1}{\|\eta\|_1} < \frac{1}{2}.
$$
Using this inequality and Definition~\ref{def.nup}, it follows that 
$\Phi$ satisfies NUP of order $s$ if and only if
\begin{equation}\label{opt.prob.1}
\max_{|S|\leq s} \ \ \sup_{\eta \in \nullsp(\Phi)\setminus \{0\}} \frac{\|\eta_S\|_1}{\|\eta\|_1} < \frac{1}{2}.
\end{equation}
Since the objective function in~\eqref{opt.prob.1} is independent of the scale of $\eta$, the condition in~\eqref{opt.prob.1} is equivalent to 
\begin{equation}\label{opt.prob.2}
\max_{|S|\leq s} \ \ \max_{\eta\in \nullsp(\Phi)\cap \mathbb{S}^{n-1}_1} \|\eta_S\|_1 < \frac{1}{2}.
\end{equation}
Since we can replace $\mathbb{S}^{n-1}_1$ with $\mathbb{B}^n_1$ (this does not change the result of the maximization problem), we find that~\eqref{opt.prob.2} is equivalent to~\eqref{eq:AltNUP}, as claimed.
\end{proof}

\begin{remark}\label{rem:nscNPhard}
The value of the left hand side of \eqref{eq:AltNUP} is called the \emph{nullspace constant} in the literature \cite{TiPf:14}. The calculation of the nullspace constant for an arbitrary matrix $\Phi$ and sparsity $s$ is known to be NP-hard \cite{TiPf:14}.
\end{remark}

The reformulation of the NUP in Lemma~\ref{lem:reformNUP} has certain benefits. 
For a fixed index set $S\subseteq \{1,\dots, n\}$, it draws our attention to the optimization problem
\begin{equation}\label{eq:l1maxccset}
\max_{\eta\in \nullsp(\Phi)\cap \mathbb{B}^{n}_1} \|\eta_S\|_1
\end{equation}
which is the maximization of a continuous convex function $\|\Gamma_S(\cdot)\|_1$ over a compact convex set $\nullsp(\Phi)\cap \mathbb{B}^{n}_1$. 
Thus, it follows from Proposition~\ref{prop:qconvmax} that the maximum is attained at an extreme point of $\nullsp(\Phi)\cap \mathbb{B}^{n}_1$. This leads us to want to understand the extreme points of this set, which can be a computationally involved task for arbitrary matrices $\Phi$. Nonetheless, one can still set a bound on the sparsity of the extreme points of $\nullsp(\Phi)\cap \mathbb{B}^{n}_1$.

\begin{lemma}\label{lem:extp}
If $\Phi\in \mathbb{R}^{m\times n}$ is an $m\times n$ matrix of rank $r$,
then extreme points of $\nullsp(\Phi)\cap \mathbb{B}_1^n$ are at most $(r+1)$-sparse. 
\end{lemma}
\begin{proof}
If $\dim(\nullsp(\Phi)) \equiv n-r\leq 1$, then the result of the lemma holds trivially since $r+1 \geq n$.

The unit $\ell_1$-sphere in $\Re^n$, namely $\mathbb{S}_{1}^{n-1}$, may be written as the union of $(n-1)$-dimensional simplices. Hence, if $\dim(\nullsp(\Phi)) \equiv n-r\geq 1$, each extreme point is contained in an $(n-1)$-dimensional simplex. 

In particular, if $\dim(\nullsp(\Phi)) \equiv n-r>1$, we argue that, no extreme point of $\nullsp(\Phi)\cap \mathbb{B}_1^n$ lies 
in the interior of these $(n-1)$-dimensional simplices. This is because in this case the intersection of $\nullsp(\Phi)$ and the interior of the $(n-1)$-dimensional simplex is either empty, or a non-singleton open convex subset of $\mathbb{R}^n$, where every point is a convex combination of two distinct points. Hence, no extreme point can live there. So, the extreme points must lie in the boundary of $(n-1)$-dimensional simplices, which are $(n-2)$-dimensional simplices. Moreover, these $(n-2)$-dimensional simplices are where one coordinate becomes zero. Hence, at least one coordinate of an extreme point living in an $(n-2)$-dimensional simplex is zero.

As long as the sum of the dimension of the simplex, which contains the extreme point, and the dimension of $\nullsp(\Phi)$ is strictly larger than $n$, we could repeat the argument in the previous paragraph. This argument stops when the dimension of the simplex containing the extreme point is $r$. Thus, at least $n-r-1$ coordinates of the extreme point have to be zero, so that  each extreme point is at most $(r+1)$-sparse.
\end{proof}

Although Lemma \ref{lem:extp} is analogous to \eqref{eq:sparkineq}, it is stronger in the sense that the statement bounds the sparsity of all the extreme points of the convex set $\nullsp(\Phi)\cap\mathbb{B}_1^n$, not just the sparsest vectors in the $\nullsp(\Phi)$, as is implied by~\eqref{eq:sparkineq}. Armed with Lemma \ref{lem:extp}, it turns out that for incidence matrices, these extreme points have a nice characterization in terms of the properties of the underlying graph. 

\begin{prop}\label{prop:extremepts}
Let $A\in \mathbb{R}^{m \times n}$ be the incidence matrix of a simple connected graph $\mathcal{G}$ that has at least one simple cycle, and let $W_1$ denote the set of normalized simple cycles of $\mathcal{G}$, i.e.
\begin{equation}\label{eq:simplecycles}
W_1 \! :=\! \left\{\frac{w(\mathcal{C})}{\| w(\mathcal{C})\|_1} \mid\;\mathcal{C} \text{ is a simple cycle of }\mathcal{G}\right\}.
\end{equation}
Then, we have
$\Ext(\nullsp(A)\cap \mathbb{B}_1^{n}) \subseteq W_1$.
\end{prop}
\begin{proof}
Let $z$ be an extreme point of $\nullsp(A)\cap \mathbb{B}_1^{n}$ with support $S$. We necessarily have $\|z\|_1=1$. Suppose that $\mathcal{G}_S$ has $\bar{m}$ vertices and $\bar{k}$ connected components. Note that $z_S\in \nullsp(A_S)\cap \mathbb{B}_1^{|S|}$ has only nonzero entries. 
We claim that $z_S$ is an extreme point of $\nullsp(A_S)\cap \mathbb{B}_1^{|S|}$. For a proof by contradiction, suppose that $z_S$ could be written as a convex combination of two distinct vectors in $\nullsp(A_S)\cap \mathbb{B}_1^{|S|}$.  In this case, one could pad those vectors with zeros at the coordinates in $S^c$ to get two distinct vectors in $\nullsp(A)\cap \mathbb{B}_1^{n}$, whose convex combination would give $z$.  Since this would contradict the fact that $z$ is an extreme point, we must conclude that $z_S$ is an extreme point of $\nullsp(A_S)\cap \mathbb{B}_1^{|S|}$. 

Using the  fact that $z_S$ is an extreme point of $\nullsp(A_S)\cap \mathbb{B}_1^{|S|}$, Lemma~\ref{lem:extp}, and Remark~\ref{remark:rank}, it follows that 
\begin{equation}\label{dim.ns.1}
|S|\leq \rank(A_S)+1 = \bar{m}-\bar{k}+1.
\end{equation}
Combining~\eqref{dim.ns.1} with the rank-nullity theorem yields
$$
\dim(\nullsp(A_S)) = |S|-\bar{m} +\bar{k}\leq 1.
$$
This may be combined with $0\neq z_S\in \nullsp(A_S)$ to conclude that $\dim(\nullsp(A_S)) = 1$; thus $\nullsp(A_S)$ is spanned by $z_S$. It now follows from Remark~\ref{remark.cs.is.ns} that $z_S$ corresponds to the unique simple cycle in $\mathcal{G}_S$. In addition, since $z_S$ has no zero entries, the corresponding simple cycle includes every edge of $\mathcal{G}_S$. It follows that $z$ corresponds to a simple cycle in $\mathcal{G}$, which combined with $\|z\|_1=1$ proves $z \in W_1$.
\end{proof}

Next we show that Proposition \ref{prop:extremepts} provides a way to study the maximization problem \eqref{eq:l1maxccset}.

\begin{thm}\label{thm:strspar} Let $A \in \mathbb{R}^{m\times n}$ be the incidence matrix of a simple, connected graph $\mathcal{G}$, and let $W_1$ denote the set of normalized simple cycles of $\mathcal{G}$ as in \eqref{eq:simplecycles}. Let $S\subseteq \{1,\dots, n\}$ be a nonempty index set. Then  
\begin{equation}\label{eq:strspar}
 \quad \max_{\eta\in \nullsp(A)\cap \mathbb{B}^{n}_1} \|\eta_S\|_1= \quad  \max_{z\in W_1} \|z_S\|_1.
\end{equation}
\end{thm}
\begin{proof}
Since $ W_1 \subset \nullsp(A)\cap \mathbb{B}^{n}_1$, we obviously have
$$ \max_{\eta\in \nullsp(A)\cap \mathbb{B}^{n}_1} \|\eta_S\|_1 \geq \quad  \max_{z\in W_1} \|z_S\|_1.$$
For the converse inequality we argue as follows: Since $\| \Gamma_S(\cdot) \|_1$ is a continuous convex function (thus also quasi-convex), and $\nullsp(A)\cap \mathbb{B}_1^n$ is a compact convex set, the maximum in the left hand side of~\eqref{eq:strspar} will be attained at an extreme point of $\nullsp(A)\cap \mathbb{B}_1^n$ (see Proposition~\ref{prop:qconvmax}). That is,
\begin{equation}\label{eq:attains_at_ext}
\max_{\eta\in \nullsp(A)\cap \mathbb{B}^{n}_1} \|\eta_S\|_1 =  \max_{\nu\in \Ext(\nullsp(A)\cap \mathbb{B}^{n}_1)} \|\nu_S\|_1.
\end{equation}
But $\Ext(\nullsp(A)\cap \mathbb{B}_1^n) \subseteq W_1$ by Proposition~\ref{prop:extremepts}. Hence,
$$ \max_{\nu\in \Ext(\nullsp(A)\cap \mathbb{B}_1^n) } \|\nu_S\|_1 \leq \quad  \max_{z\in W_1} \|z_S \|_1.$$
Combining this with \eqref{eq:attains_at_ext} we get the converse inequality, and hence the result.
\end{proof}

Theorem~\ref{thm:strspar} builds a connection between the algebraic condition NUP and a topological property of the graph, namely its simple cycles. This connection is what we will primarily exploit in the rest of the paper. 

Let us make a simple but important observation.
\begin{lemma}\label{lem:1normGammaSx} Let $\mathcal{G}$ and  $W_1$ be as in Theorem~\ref{thm:strspar}. Then, for any index set $S$ and $z\in W_1$, it holds that
$$\|z_S \|_1 = \frac{|S\cap \supprt(z)|}{\|z\|_0}.$$
\end{lemma}
\begin{proof}
For any simple cycle $\mathcal{C}$, the entries of $w(\mathcal{C})$ are in the set $\{0,-1,1\}$. Hence, entries of any $z\in W_1$ are from the set $\{0,\frac{1}{\|z\|_0},\frac{-1}{\|z\|_0}\}$, from which the result follows. 
\end{proof}

\subsection{Polynomial Time Guarantees for Sparse Recovery}
An answer to the first question posed at the beginning of Section~\ref{sect:mainresults} is given by Theorem~\ref{thm:girthsparsity}.

\begin{thm}\label{thm:girthsparsity}
Let $A \in \mathbb{R}^{m\times n}$ be the incidence matrix of a simple, connected graph $\mathcal{G}$ with girth $g$. Then, every $s$-sparse vector $\bar{x}\in \mathbb{R}^n$  is the unique solution to
$$
\min_{A\bar{x}=Ax} \|x\|_1
$$
if and only if $s < \frac{g}{2}$.
\end{thm}
\begin{proof}
From  Lemma \ref{lem:reformNUP} and Theorem \ref{thm:strspar} it follows that the NUP of order $s$ is satisfied if and only if
\begin{equation}\label{eq:mm.half}
\max_{|S|\leq s}\:\max_{z\in W_1} \|z_S\|_1<\frac{1}{2}.
\end{equation}
From Lemma~\ref{lem:1normGammaSx}, the maximum on the left of \eqref{eq:mm.half} is equal to
\begin{equation}\label{eq:mm.half.left.equiv}
\max_{|S|\leq s}\:\max_{z\in W_1} \frac{|S\cap \supprt(z)|}{\|z\|_0}.
\end{equation}
If $s\geq g$, then \eqref{eq:mm.half.left.equiv} is equal to $1$. If $s < g$, the maximum is attained at the smallest simple cycle when $S$ picks edges only from this cycle, and the maximum is equal to $\frac{s}{g}$. Hence,
\begin{equation}\label{eq.sdg}
\max_{|S|\leq s}\:\max_{z\in W_1} \|z_S\|_1= \min\{s/g,1\}.
\end{equation}
The desired result now follows from~\eqref{eq:mm.half} and~\eqref{eq.sdg}.
\end{proof}


\begin{remark}
As was mentioned in Remark~\ref{rem:nscNPhard}, calculating the nullspace constant is NP-hard in general. However, there are algorithms that can calculate the girth of a graph exactly in $\mathcal{O}(mn)$ or sufficiently accurately for our purposes in $\mathcal{O}(m^2)$ (see~\cite{ItaiR:78}). Therefore, Theorem~\ref{thm:girthsparsity} reveals that the nullspace constant can be calculated---hence the NUP can be verified---for graph incidence matrices in polynomial time. 
\end{remark}


\subsection{Guarantees for Support-dependent Sparse Recovery}

The sparse recovery result given by Theorem \ref{thm:girthsparsity} depends on the cardinality $s$ of the support of the unknown signal. For a graph with small girth (e.g., $g=3$), this result establishes sparse recovery guarantees only for signals with low sparsity levels (specifically, $s < g/2$). It is natural, therefore, to ask the following question: Given a graph with relatively small girth, can we identify a class of signals with sparsity $s \geq g/2$ that can be recovered? This leads us to the study of support-dependent sparse recovery, where we not only focus on the cardinality of the support of the unknown signal $\bar{x}$, but also take the location of the support into account.
More precisely, we have the following result.

\begin{thm}\label{thm:CheckCycle}
Let $A \in \mathbb{R}^{m\times n}$ be the incidence matrix of a simple, connected graph $\mathcal{G}$, let $W_1$ denote the set of normalized simple cycles of $\mathcal{G}$ as in \eqref{eq:simplecycles}, and let $\emptyset \neq S\subseteq \{1,\dots, n\}$. Then,  every vector $\bar{x}\in \mathbb{R}^n$ with $\supprt(\bar{x})\subseteq S$ is the unique solution to the optimization problem
\begin{equation}\label{eq:l1InThm}
\min_{A\bar{x}=Ax} \|x\|_1
\end{equation}
if and only if
\begin{equation}\label{eq:LocDepRecoCond}
\max_{z\in W_1} \frac{|S\cap \supprt(z)|}{\|z\|_0} < \frac{1}{2}.
\end{equation}
\end{thm}

\begin{proof}
From arguments in the proof of Lemma~\ref{lem:reformNUP}, we know that \eqref{eq:ImpliedIneq} is equivalent to
$$\max_{\eta\in \nullsp(A)\cap \mathbb{B}^{n}_1} \|\eta_S\|_1 < \frac{1}{2},$$
which with Theorem~\ref{thm:strspar} and Lemma~\ref{lem:1normGammaSx} gives the result.
\end{proof}

\begin{remark}
In comparison to Theorem~\ref{thm:girthsparsity}, Theorem \ref{thm:CheckCycle} provides a deeper insight into the vectors $\bar{x}$ that can be recovered from the observations $A\bar{x}$ via $\ell_1$-minimization. Specifically, any vector whose support within each simple cycle has size strictly less than half the size of that cycle, can be successfully recovered.

As an example consider the graph in Fig.~1, which has girth $g=4$. Thus, it follows from Theorem~\ref{thm:girthsparsity} that any $1$-sparse signal can be recovered. However, in reality, some signals that are supported on more than one edge can also be recovered. For instance, Theorem~\ref{thm:CheckCycle} guarantees that a signal $\bar{x}$ supported on edges $\{2, 7\}$ or $\{ 4, 6, 8\}$ can be recovered as the unique solution to the $\ell_1$ problem~\eqref{eq:l1InThm} since for each of the three simple cycles, its intersection with the support of the unknown vector $\bar{x}$ is less than  half of the length of the simple cycle. This reveals that different sparsity patterns of the same sparsity level can have different recovery performances.
\end{remark}

\begin{remark}
To use Theorem \ref{thm:CheckCycle} as a way of providing a support-dependent recovery guarantee, one needs to compute the intersection of the support of $\bar{x}$ with all simple cycles in the graph. A number of algorithms with a polynomial time complexity bound for cycle enumeration are known~\cite{mateti1976algorithms}. 
\end{remark}


\subsection{Using Measurement Sparsity to Aid Recovery}

When there is additional information about the support of the unknown signal, Theorem \ref{thm:CheckCycle} gives a necessary and sufficient condition for exact recovery. In practice, this information is usually missing. However, the special structure of the incidence matrix and its connection to the graph can help us circumvent this difficulty. Notice that the columns of any incidence matrix are always $2$-sparse, which means that the measurement $y := A\bar{x}$ will be $2s$-sparse for any $s$-sparse signal $\bar{x}$. 
Therefore, one can seek to obtain a superset of $\supprt(\bar x)$ by observing $\supprt(y)$, i.e. the vertices with nonzero measurements. Typically, this observation reduces the size of the problem, and gives rise to Algorithm~\ref{alg:recoveralg}.



\begin{algorithm}[H] 
\caption{Recovering the unknown signal by using sparse vertex measurements.}
\label{alg:recoveralg}
\begin{algorithmic}[1]
\State \textbf{INPUT:} $y\in \mathbb{R}^m$, $A\in \mathbb{R}^{m \times n}$, $\mathcal{G}(V,E)$. 
\State Set $\hat{x}=0$.
\State Define $T:=\supprt(y)$.
\State Define $S := \{j\mid e_j=(v_i,v_k)\in E \text{ and }  v_i,v_k\in T\}$, the indices of the edges of $\mathcal{G}$ connecting vertices with nonzero measurements.
\State Construct $\mathcal{G}_S$, the subgraph of $\mathcal{G}$ consisting of the edges indexed by $S$. 
\State Construct the incidence matrix $A_S$ of the subgraph $\mathcal{G}_S$.
\State Solve the $\ell_1$-minimization problem 
\begin{equation*}
\tilde{x}=\argmin_{y_T = A_S x} \|x\|_1
\end{equation*}
\State Set $\hat{x}_{S}=\tilde{x}$.
\State \Return $\hat{x}\in \mathbb{R}^n$
\end{algorithmic}
\end{algorithm}

\begin{remark}
There are cases where Algorithm~\ref{alg:recoveralg} may fail. For instance, when the unknown nonzero signal $\bar{x}$ is in the nullspace of one of the rows of $A$, say $a_k$, and $\supprt(a_k)\cap \supprt(\bar{x})\neq \emptyset$ (i.e. when the non-zero values of the signal $\bar x$ cancel each other at a vertex $k$), Algorithm~\ref{alg:recoveralg} would simply ignore the corresponding vertex, and hence, the edges connected to it. This leads to a wrong subgraph $\mathcal{G}_S$ and possibly an incorrect outcome $\hat x$. Fortunately,  this case happens with zero probability when the signal has random support and random values.
\end{remark}

\begin{cor}\label{cor.random}
Let $A \in \mathbb{R}^{m\times n}$ be the incidence matrix of a simple, connected graph $\mathcal{G}$. 
Given  a random signal $\bar{x}$, let $y:=A\bar{x}$  and $\mathcal{G}_S$ be the subgraph associated with  $y$ in Algorithm~\ref{alg:recoveralg}.
Then, Algorithm~\ref{alg:recoveralg} will almost surely recover the signal $\bar{x}$ if and only if $\supprt(\bar{x})$ picks strictly less than half of the edges from each simple cycle of $\mathcal{G}_S$.
\end{cor}

The proof of Corollary~\ref{cor.random} follows from Theorem \ref{thm:CheckCycle} and the fact that $\bar{x}$ is assumed to be randomly generated.

In comparison to $\ell_1$-minimization on the original graph, Algorithm~\ref{alg:recoveralg} may have improved support-dependent recovery performance because the formation of the subgraph may eliminate  cycles; this is especially true for large-scale networks. For concreteness, let us  give a simple example.
\begin{exm}
Consider the graph depicted in Fig.~\ref{Fig1}. Let $\bar{x}$ be a random signal that is supported on the edges $\{3,4,5,6,7,8,10\}$, so that (with probability one) there are nonzero measurements on nodes $\{1,3,4,5,6,7,8,9\}$. Then, $x$ cannot be recovered using the incidence matrix of the original graph. However, the subgraph $\mathcal{G}_S$ is acyclic, as shown in Fig.~\ref{Fig2}, which allows exact recovery via Algorithm~\ref{alg:recoveralg}.
\end{exm}
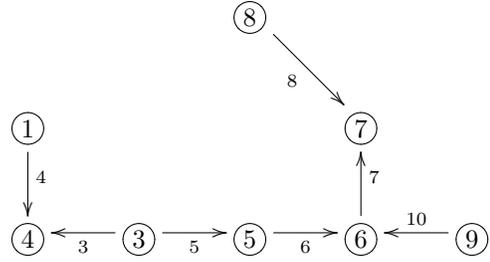
\begin{figure}[t] 
 \[ \xymatrix{
 & & \xy*{8}*\cir<6pt>{}\endxy  \ar[dr]_8 \\
 \xy*{1}*\cir<6pt>{}\endxy  \ar[d]^4 & & & \xy*{7}*\cir<6pt>{}\endxy \\
 \xy*{4}*\cir<6pt>{}\endxy  & \xy*{3}*\cir<6pt>{}\endxy \ar[l]^3  \ar[r]_5 &  \xy*{5}*\cir<6pt>{}\endxy \ar[r]_6   & \xy*{6}*\cir<6pt>{}\endxy \ar[u]_7  &  \xy*{9}*\cir<6pt>{}\endxy \ar[l]_{10}
 }
  \]
\caption{Acyclic subgraph $\mathcal{G}_S$ induced by an unknown signal with support $S=\{3,4,5,6,7,8,10\}$. }
\label{Fig2}
\end{figure}

\begin{figure}[b]
\[ \xymatrix{
   & \xy*{2}*\cir<6pt>{}\endxy \ar[dd] \ar@/_/[dl] 
   \\ 
\xy*{3}*\cir<6pt>{}\endxy \ar@/_/@{.>}[dr] 
&  &\xy*{1}*\cir<6pt>{}\endxy  \ar@/_/[ul]
\\
 &  \xy*{l+1}*\cir<12pt>{}\endxy  
 \ar@/_/[ur]
}
\]
\caption{Graph used in Experiment 1, which consists of a cycle of length 3 and a cycle of varying length $l\in\{3,5,7,9,11\}$.} \label{2Cycleplot}
\end{figure}
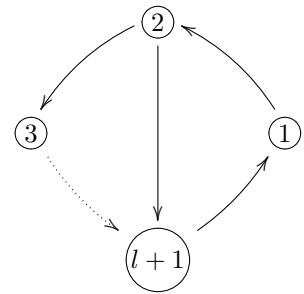


\section{Simulations}

In this section we provide numerical simulation results on the recovery performance associated with incidence matrices. In all experiments, the signal $\bar{x}$ has random support with each nonzero entry drawn i.i.d. from the standard normal distribution. We use the CVX software package \cite{grant2008cvx} to solve the optimization problems. A vector is declared to be recovered if the $2$-norm error is less than or equal to $10^{-5}$.
\begin{experiment}
Fig.~\ref{Exp1} shows the probability of exact recovery of signals via $\ell_1$-relaxation over a sequence of graphs containing two cycles (see Fig.~\ref{2Cycleplot}). Each graph has a cycle of length $3$ and a larger one of varying lengths $3, 5, 7, 9$, and $11$, respectively.  For each sparsity level, 1000 trials were performed. According to Theorem \ref{thm:girthsparsity}, for all graphs we can recover 1-sparse signals since the girth is 3 for each graph in the sequence. When the sparsity level is increased, we expect that the probability of exact recovery will increase for graphs with larger cycles because it becomes less likely that the support of the random signal will consist of more than half of the edges for one of the simple cycles in the graph. Note that this agrees with our observation in Fig.~\ref{Exp1}.
\end{experiment}
\begin{figure}[t]
\centering\includegraphics[width=0.5\textwidth]{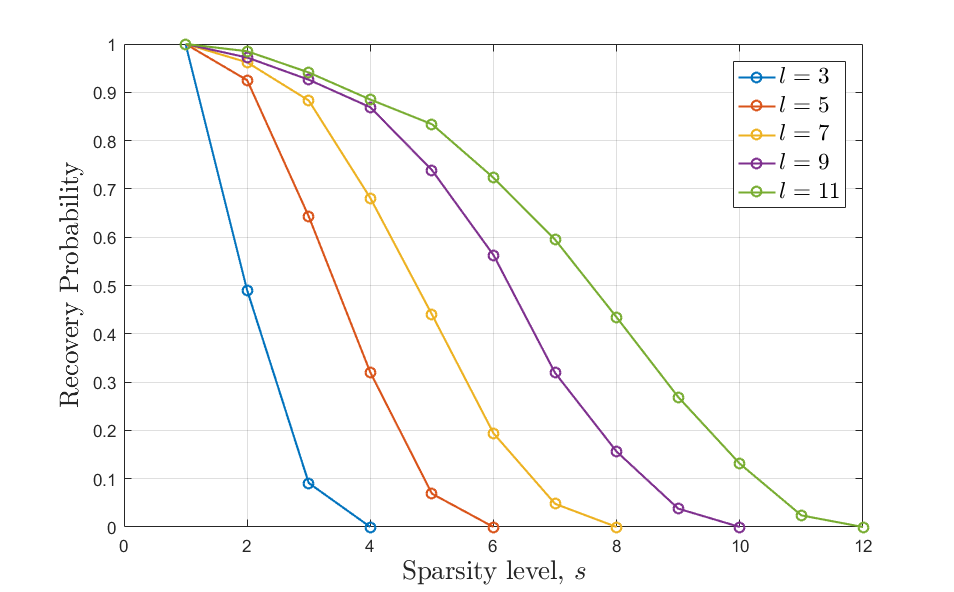}
\caption{Probability of exact recovery as a function of the sparsity level for the graphs of various loop sizes in Fig.~\ref{2Cycleplot}.}   \label{Exp1}
\end{figure}


\begin{experiment}
We now evaluate the performance of Algorithm 1 against the $\ell_1$-minimization method in (\ref{eq:mainl1min}) on two graphs with $20$ nodes (see Fig.~\ref{networkfig}). The first graph, $\mathcal{G}_B$, is a simple cycle connecting node 1 to node 20 in order (blue edges in Fig.~\ref{networkfig}). The second graph, $\mathcal{G}_{BR}$, consists of both blue and red edges. The red edges connect each node to its third neighbor, i.e., $(1,4), (2,5) \cdots (17,20)\cdots(20,3)$. Fig.~\ref{Alg20} and Fig.~\ref{Alg40} show the performance of both algorithms on $\mathcal{G}_B$ and $\mathcal{G}_{BR}$ respectively. For each sparsity level, the experiments are repeated 1000 times to compute recovery probabilities. For $\mathcal{G}_B$, Algorithm 1 outperforms $\ell_1$-minimization since the reduced graph in Algorithm 1 will be acyclic in some cases. For $\mathcal{G}_{BR}$, Algorithm 1 will eliminate small cycles when forming the subgraph and lead to higher recovery probability for fixed $s$ due to a larger girth. 
\end{experiment}
\begin{figure}[t]
\centering
\includegraphics[width=0.4\textwidth]{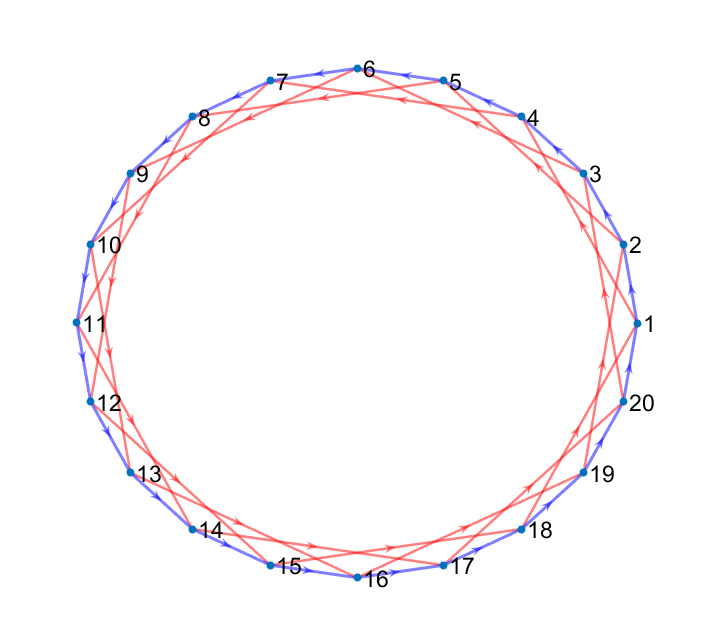}
\caption{Graphs with 20 nodes for Experiment~2. $\mathcal{G}_B$ consists of blue edges only and $\mathcal{G}_{BR}$ consists of blue and red edges.}
\label{networkfig}
\end{figure}

\begin{figure}[t]
\centering
\includegraphics[width=0.5\textwidth]{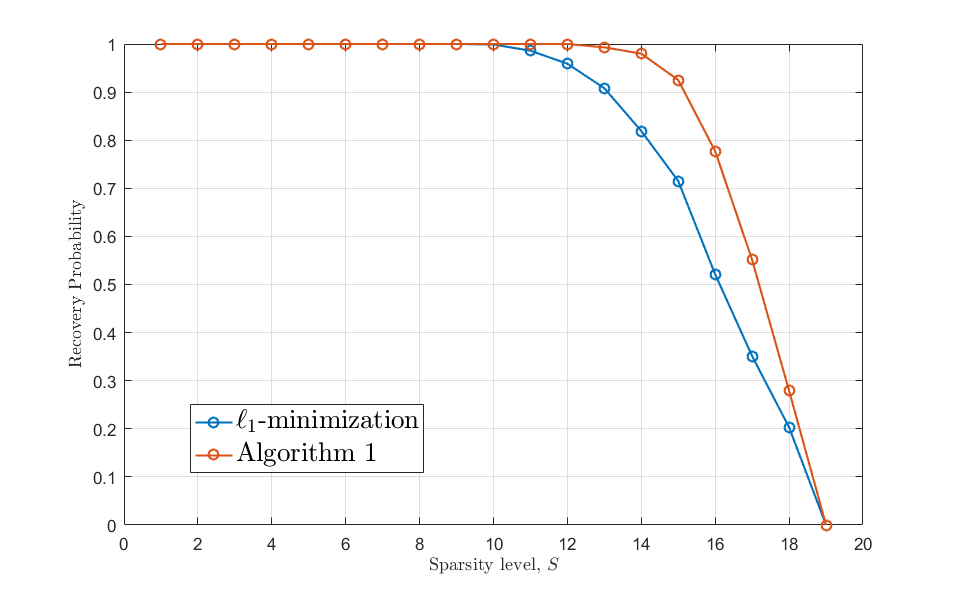}
\caption{Probability of exact recovery as a function of the sparsity level for Algorithm 1 and $\ell_1$-minimization on $\mathcal{G}_B$.} \label{Alg20}
\end{figure}

\begin{figure}[h!]
\centering
\includegraphics[width=0.5\textwidth]{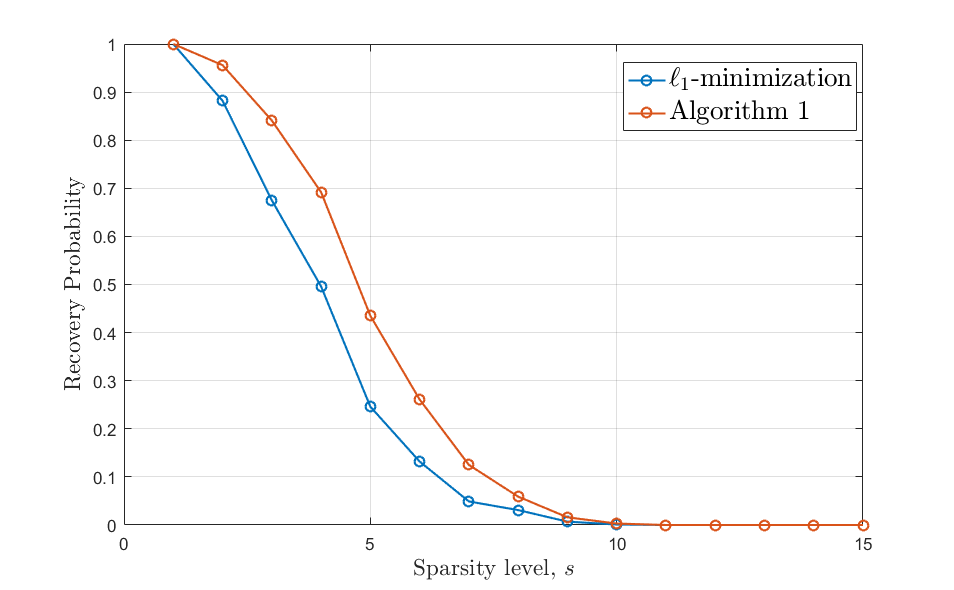}
\caption{Probability of exact recovery as a function of the sparsity level for Algorithm 1 and $\ell_1$-minimization on $\mathcal{G}_{BR}$.} \label{Alg40}
\end{figure}

\section{Conclusion}
In this paper we studied sparse recovery for the class of graph incidence matrices. For  such matrices we provided a characterization of the NUP in terms of the topological properties of the underlying graph. This characterization allows one to verify sparse recovery guarantees in polynomial time for the class of incidence matrices. Moreover, we showed that support-dependent recovery performance can also be analyzed in terms of the simple cycles of the graph. Finally, by exploiting the structure of incidence matrices and sparsity of the measurements, we proposed an efficient algorithm for recovery and evaluated its performance numerically.






\addtolength{\textheight}{-3cm}   


\bibliographystyle{ieeetr}
\bibliography{biblio}







\end{document}